\documentclass[twocolumn, pra, superscriptaddress]{revtex4-1}

\usepackage{graphicx}

\usepackage{amsmath}
\usepackage{amssymb}
\usepackage{amsthm}
\usepackage[usenames,dvipsnames]{color}
\usepackage[
 verbose,
 colorlinks=true,
 naturalnames=true,
 linkcolor=blue,
]{hyperref}
\usepackage{subfigure}
\usepackage{braket}

\usepackage{hyperref}
\hypersetup{colorlinks,linkcolor={blue},citecolor={blue},urlcolor={blue}}

  \newtheorem{theorem}{Theorem}
	\newtheorem{lemma}{Lemma}
    \newtheorem{definition}{Definition}
	\newtheorem{corollary}{Corollary}

	\def\be{\begin{equation}}
	  \def\ee{\end{equation}}
	  \def\ba{\begin{eqnarray}}
	  \def\ea{\end{eqnarray}}

\graphicspath{{"../Quantum Thermometry/Paper/"}}
	  
\begin{document}
\date{\today}
\title{Nonclassical correlations for quantum metrology in thermal equilibrium}

\author{Akira Sone}
%\email{}
\affiliation{Research Laboratory of Electronics, Massachusetts Institute of Technology, Cambridge, MA 02139 }
\affiliation{Department of Nuclear Science and Engineering, Massachusetts Institute of Technology, Cambridge, MA 02139}

\author{Quntao Zhuang}
%\email{}
\affiliation{Research Laboratory of Electronics, Massachusetts Institute of Technology, Cambridge, MA 02139 }
\affiliation{Department of Physics, Massachusetts Institute of Technology, Cambridge, MA 02139}
\affiliation{Department of Physics, University of California Berkeley, Berkeley, California 94720, USA}

\author{Changhao Li}
%\email{}
\affiliation{Research Laboratory of Electronics, Massachusetts Institute of Technology, Cambridge, MA 02139 }
\affiliation{Department of Nuclear Science and Engineering, Massachusetts Institute of Technology, Cambridge, MA 02139}

\author{Yi-Xiang Liu}
%\email{}
\affiliation{Research Laboratory of Electronics, Massachusetts Institute of Technology, Cambridge, MA 02139 }
\affiliation{Department of Nuclear Science and Engineering, Massachusetts Institute of Technology, Cambridge, MA 02139}

\author{Paola Cappellaro}
\email{pcappell@mit.edu}
\affiliation{Research Laboratory of Electronics, Massachusetts Institute of Technology, Cambridge, MA 02139 }
\affiliation{Department of Nuclear Science and Engineering, Massachusetts Institute of Technology, Cambridge, MA 02139}

\begin{abstract}
Nonclassical correlations beyond entanglement might provide a resource in quantum information tasks, such as quantum computation or quantum metrology. 
Quantum discord is a measure of nonclassical correlations, to which entanglement belongs as a subset. Exploring the operational meaning of quantum discord as a resource in quantum information processing tasks, such as quantum metrology, is of essential importance to our understanding of nonclassical correlations. In our recent work [Phys. Rev. A, \textbf{98}, 012115 (2018)], we considered a protocol---which we call the greedy local thermometry protocol--- for estimating the temperature of thermal equilibrium states from local measurements, elucidating the role of diagonal discord   in enhancing  the protocol sensitivity in the high-temperature limit.  
In this paper, we extend our results to a general greedy local parameter estimation scenario. In particular, we introduce a quantum discord---which we call discord for local metrology---to quantify the nonclassical correlations induced by the local optimal measurement on the subsystem. We demonstrate explicitly that discord for local metrology plays a role in sensitivity enhancement in the high-temperature limit by showing its relation to loss in quantum Fisher information. In particular, it coincides with diagonal discord for estimating a linear coupling parameter. 
\end{abstract}

\maketitle

%todo: Comparison to Modi PRX?

\section{Introduction}
Although the ability of entanglement to enhance quantum metrology has been well explored in ideal scenarios~\cite{giovannetti2006, giovannetti2011advances},  experimental constraints, such as noise, mixed states, and restriction to local measurements, usually make reaching the ultimate quantum limit impossible. In this context, a more general study of the role of {nonclassical correlations} in quantum metrology is critical as it can lead to more general measurement schemes, such as quantum illumination~\cite{Weedbrook16}, that take advantage of nonclassical properties~\cite{Tilma10}. 

{Nonclassical} correlations described by the quantum discord are of particular relevance as they quantify loss of information as a result of measuring a local subsystem~\cite{Girolami13, Olivier01} and can be applied to mixed states. The role of discordlike correlations has thus been recently studied in the context of parameter estimation~\cite{Braun18}, such as the geometric discord in phase estimation~\cite{Kim17}, quantum discord {in the global phase estimation with mixed states~\cite{Modi11, Modi16, Modi18a, Modi18b}}, in local phase estimation assisted by interferometry~\cite{Girolami14, Girolami15}, and  the diagonal discord in quantum thermometry~\cite{Sone18a}. 

Most of these works have analyzed the usual scenario for quantum parameter estimation, where a quantum (entangled) probe evolves under the action of an Hamiltonian that depends on the external parameter to be measured, before a measurement is performed on the final state~\cite{giovannetti2006, giovannetti2011advances, Degen16x}. 
{Although the optimal measurement does not require global measurements on the total system for schemes without entanglement~\cite{giovannetti2006},  for the entanglement-enhanced schemes described above, a global measurement is usually needed to achieve the optimal performance~\cite{Rafal14}.}

Since performing a global measurement is usually a demanding task, and one has to rely on local adaptive measurements, it is important to study whether this restriction degrades the achievable estimation performance in the case of nonclassical correlations more general than entanglement. To  better focus on this question, we consider a different metrology scenario where the parameter is not encoded during the evolution but in the equilibrium state. We show that, for a local detection protocol, nonclassical correlations in the state can be detrimental in contrast to the dynamic scenario where they help in the estimation. In particular, we consider a ``greedy" local measurement scheme~\cite{Sone18a} in which each subsystem is measured sequentially with a local optimal measurement for estimating a general parameter~(see Fig.~\ref{fig:idea2}). This protocol belongs to the class of local operations and classical communication (LOCC)~\cite{nielsen1999conditions}. In addition, we focus on  systems at thermal equilibrium in the Gibbs state and consider the  high-temperature limit, which is a practical scenario in various systems, such as a  room-temperature NMR system or biological system, and where only nonclassical correlations beyond entanglement are typically found.
Even in this regime, we find a precision loss when considering only local measurements, and we bound it by considering the discord present in the system.

Hamiltonian parameter estimation at thermal equilibrium has been considered before by Mehboudi \textit{et al.} ~\cite{Mehboudi16} in which they considered a special Hamiltonian consisting of two commuting operators to which temperature-independent parameters are linearly coupled. For this special case, they proved that the quantum Fisher information (QFI) for estimating either parameter can be characterized as a curvature of the Helmholtz free energy at an arbitrary temperature.  However, for a general Hamiltonian $H_\lambda$ parametrized by a temperature-independent parameter $\lambda$, this is not always the case because of the noncommutativity of the Hamiltonian and the generator of parameter $\lambda$. Still, in the high-temperature limit, the QFI can be well approximated by the susceptibility as discussed in Sec.~\ref{sec:globallocal}, and we can apply the relation provided in Ref.~\cite{Mehboudi16}.

\begin{figure}[hbt]
\subfigure{
\includegraphics[width=0.45\textwidth]{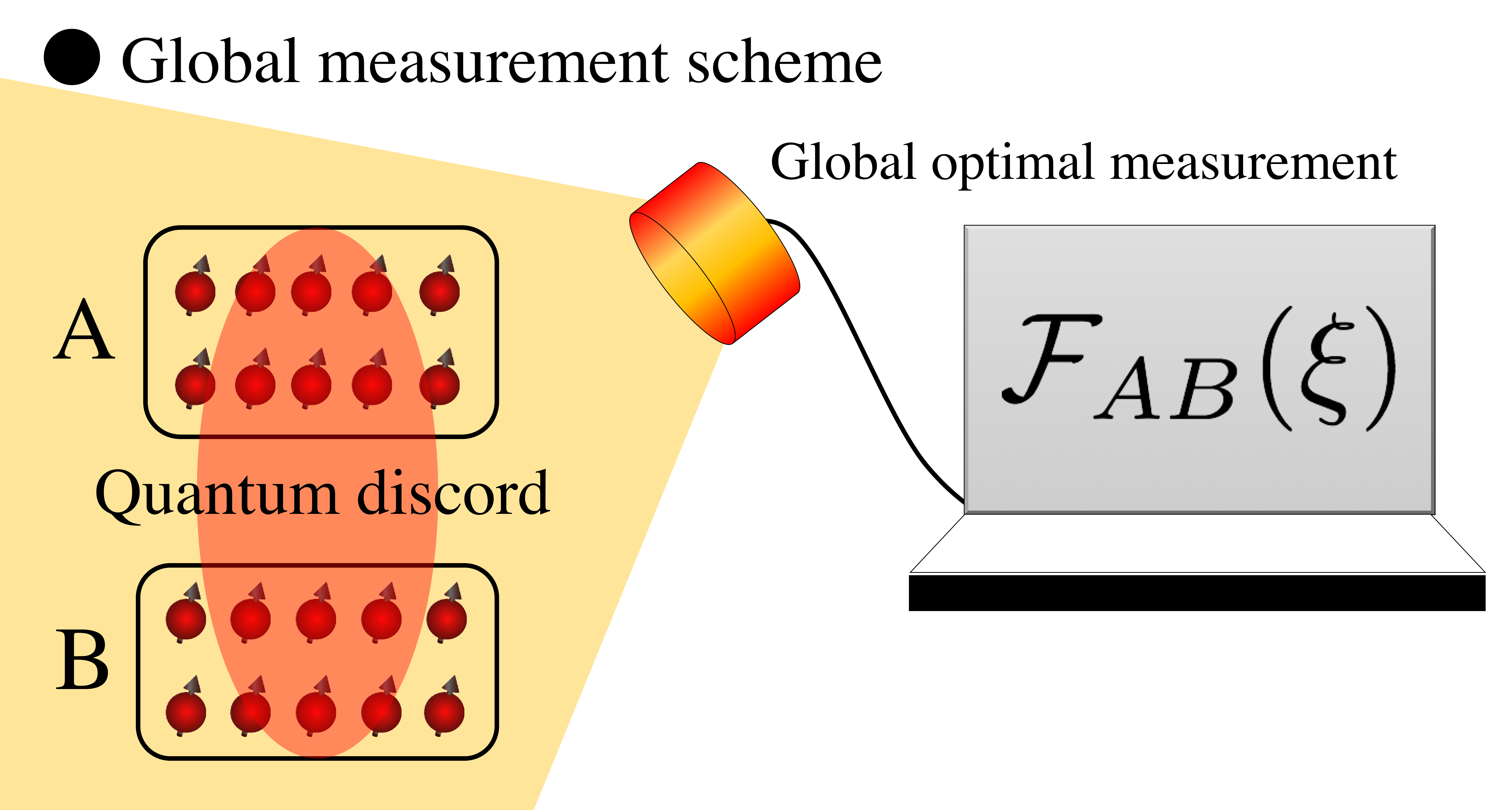}
}
\subfigure{
\includegraphics[width=0.45\textwidth]{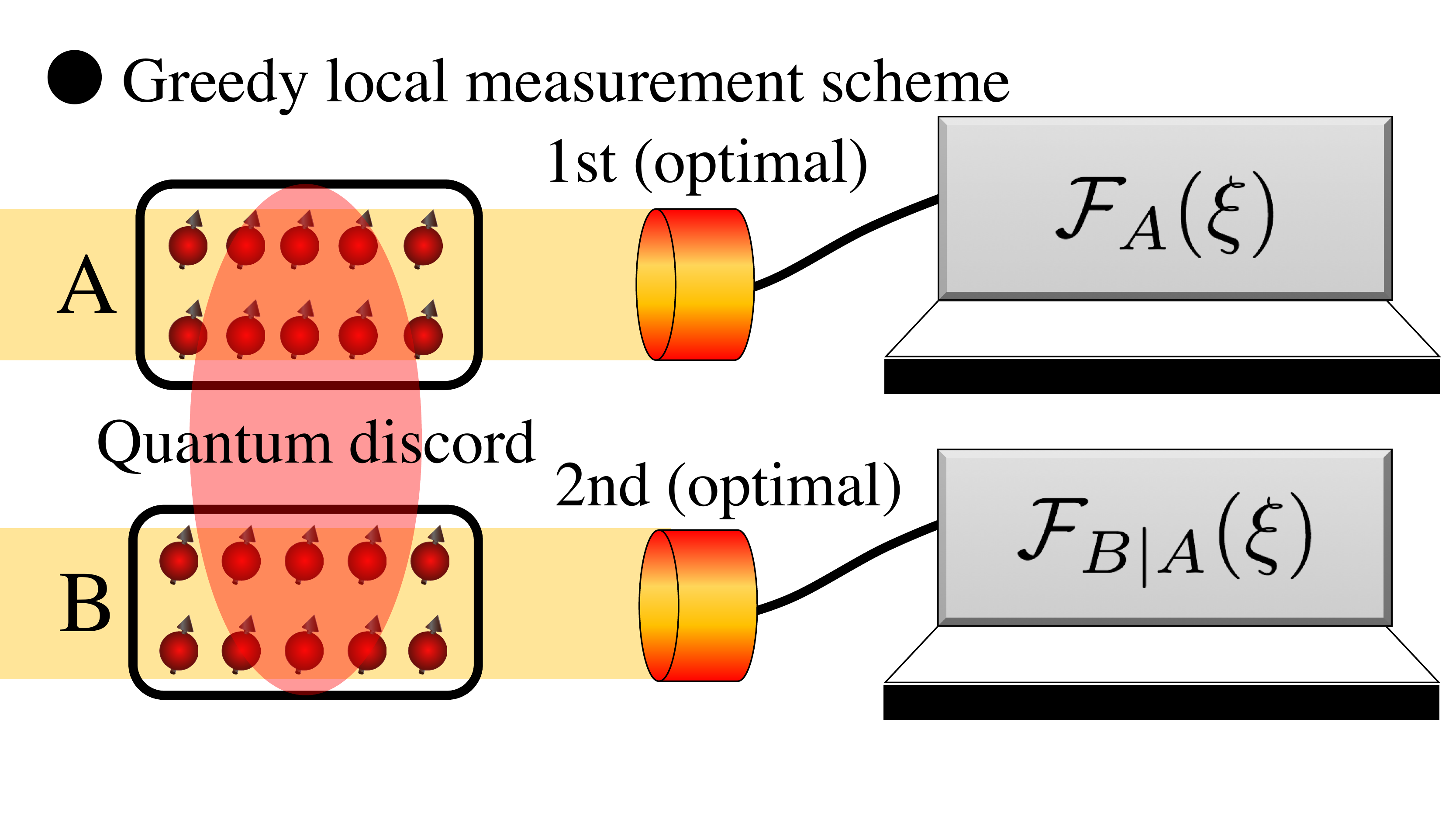}
}
\caption{Global measurement and greedy local measurement scheme: One first measures a subsystem $A$ with local optimal measurement in the sense of the local QFI and then measure the other subsystem $B$ in order to estimate an unknown parameter $\xi$. The constrained QFI is given as $\mathcal{F}_{A\to B}(\xi)=\mathcal{F}_{A}(\xi)+\mathcal{F}_{B|A}(\xi)$. We explore the relation between  the quantum discord $D_{A\to B}(\xi)$ and the precision loss $\Delta\mathcal{F}(\xi)=\mathcal{F}_{AB}(\xi)-\mathcal{F}_{A\to B}(\xi)$.
}
\label{fig:idea2}
\end{figure}

The paper is organized as follow. In Sec.~\ref{sec:QFI}, we review the QFI for estimating a single parameter and discuss the QFI in the global measurement scheme, namely, global QFI, in Sec.~\ref{sec:fisher}, and the constrained QFI in the greedy local measurement scheme, namely, LOCC QFI, in Sec.~\ref{sec:LOCC}. Based on the definition of quantum discord~\cite{Olivier01}, we introduce a quantum discord induced by local optimal measurements by considering the greedy local measurement scheme, namely, \textit{discord for local metrology} in Sec.~\ref{sec:qd}. Then, we show the relation between the discord for local metrology and precision loss quantified by the difference between global QFI and LOCC QFI at high temperatures in Sec.~\ref{sec:diagonal} and demonstrate that discord for local metrology coincides with diagonal discord when the parameter to be estimated is linearly coupled. 
Before concluding, we also provide examples to further illustrate our results.
%In Sec.~\ref{sec:example}
%we further provide examples of two-qubit Heisenberg interaction model to verify our result by focusing on thermometry~(Sec,~\ref{sec:thermometry}), coupling strength estimation~(Sec.~\ref{sec:coupling}), and magnetometry~(Sec.~\ref{sec:magnetic}), before concluding our results in Sec.~\ref{sec:conc}. 

\section{Global and greedy local measurement scheme}
\label{sec:globallocal}
We first review the definition of  QFI for estimating a single parameter, and discuss QFI for  global and  local measurement schemes. In particular, we devise an {optimal} measurement protocol that only exploits local measurements and define an associated QFI metric to evaluate its performance. 
%todo: can we say optimal??

\subsection{QFI for estimating a single parameter}
\label{sec:QFI}
The ultimate precision of parameter estimation is quantified by the QFI. 
Let $\xi$ be the parameter to be estimated, which could be a temperature independent parameter $\lambda$ in the Hamiltonian $H_\lambda$ or the temperature $T$ itself, i.e., $\xi\in\{\lambda, T\}$. 
Although often  $\xi$ is   estimated from a state $\rho_\xi$ that arises after interacting with the external field to be measured for a given time, here we consider a different scenario, where $\rho_\xi$ is an equilibrium state that is determined by the parameter-dependent Hamiltonian.
 The variance $(\delta\xi)^2$ quantifies the estimate precision. Its lower bound, which is the ultimate precision limit achievable, is  bounded by the quantum Cram\'{e}r-Rao bound $(\delta\xi)^2\ge1/\mathcal{F}(\xi,\rho_{\xi})$~\cite{Helstrom_1976, Holevo_1982, Yuen_1973}. 
Here, $\mathcal{F}(\xi,\rho_{\xi})$ is the QFI, defined as $\mathcal{F}(\xi,\rho_{\xi})=-2\lim_{\epsilon\to0}\partial^2\epsilon\mathbb{F}[\rho_\xi,\rho_{\xi+\epsilon}]$, where $\mathbb{F}[\rho,\sigma]$ denotes the fidelity between states $\rho$ and $\sigma$~\cite{Jorza93}.

\subsection{Global QFI}
\label{sec:fisher}
Consider a finite-dimensional system described by a Hamiltonian $H_\lambda$ parametrized by a single temperature-independent parameter $\lambda$ at temperature $T$. We assume the state to be in a Gibbs state, 
%Then, in order to emphasize the parameterization of Gibbs state by two possible unknown parameter $\lambda$ or $T$, we use $\xi$ as a subscript to describe Gibbs state as 
$\rho_\xi=e^{- H_\lambda/T}/\mathcal{Z}$, where we set the Boltzmann constant to be unit  $k_B=1$,
and $\mathcal{Z}=\text{Tr}[e^{-H_\lambda/T}]$ is the partition function. 
%For simplicity, in the following, we use $\rho_\xi$.

We first consider a global measurement scheme for a finite-dimensional system and derive the relation between the global QFI $\mathcal{F}(\xi,\rho_{\xi})$ and the entropy of the global system $S(\rho_\xi)$ in the high-temperature limit. We have obtained the following lemma. 
\begin{lemma}
\label{lemma1}
Consider a finite-dimensional system
%bipartite system ($AB$) 
in Gibbs state at temperature $T$, with its Hamiltonian parametrized by a temperature-independent parameter $\lambda$ to be estimated. %Let $\xi\in\{\lambda,T\}$ be an unknown parameter to be estimated. 
%When $\mathcal{F}(\xi;T)$ is the global QFI for estimating $\xi$ and  is  Then, in the high-temperature limit $(T\gg 1)$, we have
Then, the global QFI for estimating $\lambda$ and the total system entropy, $S(\lambda;T)$ are related as
\begin{equation}
\partial_T\big(T\mathcal{F}(\lambda;T)\big)=\partial_\lambda^2 S(\lambda;T)+O(T^{-3}).
\label{eq:entropy2}
\end{equation}
\end{lemma}
{The full proof is in Appendix~\ref{app:lemma1}; here, we explain the basic idea of the proof. 
In the high-temperature limit, the QFI for estimating $\lambda$ can be quantified by the susceptibility $\chi(\lambda;T)$ to leading order: $\mathcal{F}(\lambda;T)=\chi(\lambda;T)/T+O(T^{-3})$. 
From the relation between the general susceptibility and entropy, $\partial_T \chi(\lambda;T)=\partial_\lambda^2 S(\lambda;T)$, we can obtain Eq.~(\ref{eq:entropy2}). 
Furthermore, let $A(\lambda;T)$ be the Helmholtz free energy. Then, from the relation between the Helmholtz free energy and entropy, $\partial_T A(\lambda;T)=-S(\lambda;T)$, we can obtain:
\begin{align*}
\mathcal{F}(\lambda;T)=-\frac{1}{T}\partial_\lambda^2 A(\lambda;T)+O(T^{-3}).
\end{align*}
This recovers the result of Ref.~\cite{Mehboudi16} to  leading order, which demonstrates that, in the high-temperature limit, the QFI can be characterized as the curvature of the Helmholtz free energy.} % in a good approximation for the leading order. %which is equivalent to the result of ~\cite{Mehboudi16}.
%\as{To Quntao: Besides the detailed description in Appendix.~\ref{app:lemma1}, I roughly explained what is going on, and showed that the results of Mehboudi \textit{et at} can be recovered to the leading order. This is what Paola asked me to do because she wanted more description about the relation between QFI and susceptibility as well as the Helmholtz free energy. So I added. Please modify if needed.  }
If the parameter to be estimated is the temperature $\xi=T$, the relation becomes \textit{exact}:
%todo: is it exact, or there's an O(T-4)??
%\pc{is it exact, or there's an O(T-4)??}
%\as{Yes. for the temperature estimation, it is exact.}
\begin{corollary}
For a system in the Gibbs state, we have 
\begin{equation}
\partial_T\big(T\mathcal{F}(\lambda;T)\big)=\partial_T^2 S(T).
\label{eq:entropy3}
\end{equation}
\end{corollary}
%where $O(T^{-\alpha_T})=0$, i.e. when estimating temperature the relation becomes exact, and $\alpha_\lambda=3$. 

In the classical case, Eq.~(\ref{eq:entropy2}) becomes exact, as it can also be derived from  properties of the classical Fisher information in the linear exponential family~\cite{MacKay,CoverThomas,Amari}.

Let us define the optimal measurement in the high-temperature limit as the measurement which achieves the ultimate precision up to  order $O(T^{-2})$ of the QFI (for thermometry $O(T^{-4})$ of the QFI~\cite{Sone18a}).  
Different from  the thermometry case ($\xi=T$), 
whereas to estimate a generic parameter $\lambda$ the optimal measurement is generally not the projection measurement onto  energy eigenstates, this is instead the case for thermometry or if $\lambda$ is linearly coupled to the Hamiltonian.   
Formally, we have the following lemma (see Appendix~\ref{app:lemma2} for proof):
\begin{lemma}
\label{lemma2}
Consider a finite-dimensional system in the Gibbs state at temperature $T$ with its Hamiltonian parametrized by a temperature-independent parameter $\lambda$ to be estimated. If the Hamiltonian depends only linearly on $\lambda$, 
%Let $\xi\in\{\lambda,T\}$ be an unknown parameter to be estimated. When $\xi=T$ or $\xi=\lambda$ is a linear coupling parameter of the Hamiltonian $H_\lambda$ to be estimated, 
i.e., $\partial_\lambda^2 H_\lambda=0$,  projection measurements on the energy eigenstates are optimal to estimate $\lambda$.
\end{lemma}
\begin{corollary}
Since the temperature multiplies the Hamiltonian in the Gibbs state, projection on the energy eigenstates is also optimal for thermometry.
\end{corollary}
%todo: but actually, since it enters in the partition function, this is not very precisely said... 

Here, we note that, for a generic Hamiltonian $H_\lambda$, the susceptibility with respect to $\lambda$ is given by
\begin{align*}
\begin{split}
    \chi(\lambda;T)&=\frac{\langle G_\lambda^2\rangle-\langle G_\lambda\rangle^2}{T}-\langle \partial_\lambda G_\lambda\rangle\\
    &=\frac{(\delta G_\lambda)^2}{T}-\langle \partial_\lambda G_\lambda\rangle,
\end{split}
\end{align*}
where $G_\lambda=\partial_\lambda H_\lambda$. From Eq.~(\ref{eq:QFI}), the QFI becomes:
\begin{equation}
    \mathcal{F}(\lambda;T)=
    \frac{(\delta G_\lambda)^2}{T^2}-\frac{\langle \partial_\lambda G_\lambda\rangle}{T}+O(T^{-3}).
\label{eq:fisherG}
\end{equation}
If $\lambda$ is linearly coupled to the Hamiltonian, i.e., $\partial_\lambda G_\lambda=\partial_\lambda^2 H_\lambda=0$, the projection measurements on the energy eigenstate are optimal since measuring $G_\lambda$ corresponds to  projection measurements on the energy eigenstates 
and  the sensitivity of measuring $G_\lambda$  saturates the Fisher information as follows
\begin{equation}
(\delta \lambda)^2=\frac{(\delta G_\lambda)^2}{\left(\partial_\lambda \braket{G_\lambda}\right)^2}=\frac{(\delta G_\lambda)^2}{\chi(\lambda;T)^2}=\frac{T^2}{(\delta G_\lambda)^2}\approx\frac{1}{\mathcal{F}(\lambda;T)}.
\label{eq:diagonaloptimal}
\end{equation}
%todo: isn't this true only approximately??
%\pc{isn't this true only approximately??}
%\as{Yes. It is true only approximately.}
For a general parameter $\lambda$, we usually have $\langle \partial_\lambda G_\lambda\rangle\neq 0$, and from Eq.~(\ref{eq:fisherG}), the projection measurements on the energy eigenstate are not optimal. However, there still exists a set of observables that  achieves the optimal measurement.

\subsection{LOCC QFI}
\label{sec:LOCC}
Global measurements on a composite system are generally required to achieve the optimal QFI, but are usually difficult to implement. If only local measurements are available, even the best measurement protocol might not reach optimality. 
Here, we consider a  local measurement scheme with sequential local optimal measurements on subsystems that we call ``greedy'' local measurement scheme~\cite{Sone18a}. 
%The scheme belongs to the class of LOCC  so that the state of the subsystems are conditioned on the previous measurement result. 
%Note that the following discussion is valid for bipartite state. 
This scheme belongs to the class of LOCC, thus we call the constrained QFI of this scheme LOCC QFI. %Let us review the concept of LOCC QFI. 

Consider an arbitrary bipartite system in state $\rho^{AB}_\xi$. In the greedy local measurement scheme, we first perform a local optimal {projection} measurement $\tilde{\Pi}_j^A$  on the first subsystem, {where we use the notation $\tilde{\Pi}$ in order to emphasize that the measurement is optimal}. 
After the measurement,  the state of subsystem $B$ is a conditional state based on the measurement result of $\tilde{\Pi}_j^A$,
 $\rho^{B|\tilde{\Pi}_j^A}_\xi=\text{Tr}_A[(\tilde{\Pi}_j^A\otimes\openone^B)\rho^{AB}_\xi(\tilde{\Pi}_j^{A\dagger}\otimes\openone^B)]/p_j(\xi)$, with  $p_j(\xi)=\text{Tr}[(\tilde{\Pi}_j^A\otimes\openone^B)\rho^{AB}_\xi(\tilde{\Pi}_j^{A\dagger}\otimes\openone^B)]$ the measurement probability. Given the conditional QFI for outcome $j$, $\mathcal{F}_{B|\tilde{\Pi}_j^A}(\xi)=\mathcal{F}(\xi,\rho^{B|\tilde{\Pi}_j^A}_\xi)$, the unconditional local QFI for subsystem $B$ is given by
\begin{align*}
    \mathcal{F}_{B|A}(\xi)=\sum_j p_j(\xi)\mathcal{F}_{B|\tilde{\Pi}_j^A}(\xi),
\end{align*}
%where $\mathcal{F}_{B|\tilde{\Pi}_j^A}(\xi)$ is the conditional QFI defined as $\mathcal{F}_{B|\tilde{\Pi}_j^A}(\xi)=\mathcal{F}(\xi,\rho_{B|\Pi_j^A ,\xi})$. 
Note that  feed forward is required  as the optimal measurement on $B$ depends on the outcome of $\tilde{\Pi}_j^A$. From the additivity of the Fisher information,  the LOCC QFI $\mathcal{F}_{A\to B}(\xi)$ is given by
\begin{align*}
    \mathcal{F}_{A\to B}(\xi)=\mathcal{F}_A(\xi)+\mathcal{F}_{B|A}(\xi),
\end{align*}
where $\mathcal{F}_A(\xi)=\mathcal{F}(\xi,\rho^A_{\xi})$ is the local QFI for subsystem $A$~\cite{Lu12, Micadei15, Sone18a}. {Note that LOCC QFI has been originally proposed by Ref.~\cite{Lu12} from an information-theoretic perspective and by Ref.~\cite{Micadei15} from a quantum metrology perspective.} By definition, the global QFI $\mathcal{F}_{AB}(\xi)$ always satisfies $\mathcal{F}_{AB}(\xi)\ge \mathcal{F}_{A\to B}(\xi)$~\cite{Micadei15,Sone18a}. Here, we are interested in relating the precision loss, 
\begin{align*}
 \Delta\mathcal{F}(\xi)=\mathcal{F}_{AB}(\xi)-\mathcal{F}_{A\to B}(\xi),
\end{align*}
  due to local measurements to the presence of nonclassical correlations in $\rho^{AB}_\xi$.

\section{Discord for local metrology}
\label{sec:qd}

%1. DLM satisfies invariant under local unitary, because optimal measurement basis also changes.

%2. DLM is a function of a quantum state and a paramter---it is not a typical correlation measure for a state.

%3. DLM=0 implies CQ state, but does CQ state always imply DLM=0?

Nonclassical correlations associated with the loss of quantum certainty in local measurements have been quantified by quantum discord~\cite{Olivier01,Girolami13,Henderson01}.
For a bipartite system $(AB)$, the quantum discord~\cite{Olivier01} upon measuring  subsystem $A$  is defined as
\begin{align*}
    D_{A\to B}=-S_{AB}+S_{A}+\min_{\{\Pi_j^{A}\}}S_{B|\{\Pi_j^A\}},
\end{align*}
where $\{\Pi_j^A\}$'s are the set of projection measurements on subsystem $A$, and $S_i=-\text{Tr}[\rho_i \ln\rho_i]$ is the entropy of state $\rho_i$. {Here, $S_{B|\{\Pi_j^A\}}$ is defined as 
\begin{align*}
    S_{B|\{\Pi_j^A\}}=\sum_{j}p_j S_{B|\Pi_j^A},
\end{align*}
with $p_j=\text{Tr}[(\Pi_j^A\otimes\openone_B)\rho_{AB}(\Pi_j^{A\dagger}\otimes\openone_B))]$ the probability associated with the projection measurement $\Pi_j^A$.} The minimization over all sets of  projection measurements on subsystem $A$ is required in order for quantum discord to be basis independent and for extracting maximum information about subsystem $B$. 

In order to connect nonclassical correlations to the precision loss in metrology, we need to define a related metric, that we call  \textit{discord for local metrology} where the minimization is restricted to projectors achieving optimal estimate of $\xi$:
 
\begin{definition}
\label{definition}
Let $\{\tilde{\Pi}_j^A \}$ be a set of optimal projection measurements on subsystem $A$ so that there exists an observable $\tilde{\Gamma}^A=\sum_j c_j \tilde{\Pi}_j^A~(c_j\in\mathbb{C})$, which can achieve the ultimate precision of estimating $\xi$, i.e.,
\begin{align*}
    (\delta\xi)^2=\frac{(\delta \tilde{\Gamma}^A)^2}{\big(\partial_\xi\langle \tilde{\Gamma}^A\rangle\big)^2}=\frac{1}{\mathcal{F}_A(\xi)},
\end{align*}
where $\mathcal{F}_A(\xi)$ is the local QFI for estimating $\xi$ from  $\rho^A_{\xi}$. %of subsystem $A$. 
Then, \textit{discord for local metrology} $\tilde{D}_{A\to B}(\xi)$ is defined as
\begin{align*}
    \tilde{D}_{A\to B}(\xi)=-S_{AB}(\xi)+S_{A}(\xi)+\min_{\{\tilde{\Pi}_j^{A}\}}S_{B|\{\tilde{\Pi}_j^{A}\}}(\xi),
\end{align*}
which is minimized over all the possible sets of  projection measurements that are optimal for estimating the parameter $\xi$.
\end{definition}
{The minimization indicates that discord for local metrology is independent of the choice of the optimum basis for estimating $\xi$.} Because the measurement basis is chosen according to the optimal parameter estimation, discord for local metrology is an upper bound of the discord, i.e., $\tilde{D}_{A\to B}(\xi)\ge D_{A\to B}$. Also, the minimization is required to avoid the ambiguity when multiple projection bases are optimal. 
{Note that the discord for local metrology is a function of a state and a parameter; therefore, it is not a typical correlation measure for the state. Discord for local metrology has the following properties:
\begin{enumerate}
\item{$\tilde{D}_{A\to B}\ge 0$~(nonnegative);}
\item{$\tilde{D}_{A\to B}\neq \tilde{D}_{B\to A}$ (asymmetric);}
\item{If the total system is in the product state, i.e., $\rho_{AB}=\rho_A\otimes\rho_B$, then $\tilde{D}_{A\to B}=0$;
If $\tilde{D}_{A\to B}=0$, then the total system is in a classical-quantum state, i.e., $\rho_{AB}=\sum_j p_j\ket{j}\bra{j}\otimes \rho_{j,B}$, for some set of orthonormal basis vectors $\{\ket{j}\}$, probability distribution $\{p_j\}$ and states $\{\rho_{j,B}\}$.
}
\item{$\tilde{D}_{A\to B}$ is invariant under local unitary operations.}
\end{enumerate}
Properties (1) and (2) are trivial. The first half of Property (3) is straightforward, and the second part follows from the fact that $\tilde{D}_{A\to B}(\xi)\ge D_{A\to B}$; thus $\tilde{D}_{A\to B}(\xi)=0$ leads to zero discord, and the state must be classical quantum.
Property (4) is due to the state dependence of  the local measurement basis, which makes the quantity only a function of the state and parameter choice. Local unitary operations change the state, but the optimal basis also changes accordingly, thus leaving invariant the discord for local metrology. Note that one does not expect invariance under more general local operations, since discord can increase under local noise~\cite{Streltsov11}.
} Property (4) distinguishes our metric from the family of basis-dependent discord \cite{Yadin16, ZhouBDdiscord18, Ma16} with which it otherwise shares many commonalities. %\pc{I'm not sure what are the things in common and what is different. So I'm not sure if this revised statement is correct.}

%(see \textcolor{red}{https://journals.aps.org/prl/abstract/10.1103/PhysRevLett.107.170502}) 

{Since discord for local metrology satisfies the conditions of non-negativity and invariance under local unitary operations, we can regard it as a \textit{good} measure of correlations~\cite{LiuThesis18}. 
%Zi-Wen's PhD thesis 
%However, it has the following unpleasant properties as a discord---
Although it can be nonzero for some specific classical-quantum state, an unpleasant property for a discord metric,  it is a practical quantity to measure  correlations in terms of local optimal measurement for metrology. 
%Although discord for local metrology seems a to belong to the family of basis-dependent discord \cite{Yadin16, ZhouBDdiscord18, Ma16},
%Yadin16: PRX 6,041028 (2016)
%ZhouBDdiscord18: Hongyi Zhou, Xiao Yuan, and Xiongfeng Ma: arXiv:1808.03757
%Ma16: J. Ma, B. Yadin, D. Girolami, V. Vedral, and M. Gu, Phys. Rev. Lett. 116, 160407 (2016).
%they are different from each other for some properties. For example, discord for local metrology in invariant under local unitaries while basis-dependent discord is not. 
}
%\as{To Quntao: I have added the properties of discord for local metrology, and indicated that it could be a ``good" measure of correlations. Could you check whether the description is good or not. Please modify or add contents if needed. Especially, it might be great if we could add more detailed explanation to 
%\begin{enumerate}
%\item{why $\tilde{D}_{A\to B}$ is nonincreasing under local operations?}
%\item{why $\tilde{D}_{A\to B}$ is invariant under local unitary operations?}
%\item{why $\tilde{D}_{A\to B}=0$ cannot indicate that total system is a classical-quantum state?}
%\item{why the $\tilde{D}_{A\to B}$ could be nonzero if the total system is a classical state?}
%\end{enumerate}
%Could you help? 
%}

For a general bipartite system, it is a demanding task to find $\{\tilde{\Pi}_j^A\}$. However, when $\xi=T$ or $\xi=\lambda$ is a linear coupling parameter for a Gibbs state in the high-temperature limit, $\tilde{\Pi}_j^A$ becomes the eigenbasis of $\rho_A$, i.e., $\rho_A=\sum_{j}r_j\tilde{\Pi}_j^{A}$, as shown in Sec.~\ref{sec:fisher}. Therefore, $\tilde{D}_{A\to B}(\xi)$ becomes the so-called diagonal discord $\mathcal{D}_{A\to B}(\xi)$~\cite{Liu17}.

\section{Quantifying $\Delta\mathcal{F}(\xi; T)$ via  $\tilde{D}_{A\to B}(\xi; T)$}
\label{sec:diagonal}
In this section, we prove our main result, Theorem~\ref{theorem}, stating the relation between the discord for local metrology and the precision loss quantified by the difference between global QFI and LOCC QFI. 

\begin{theorem}
\label{theorem}
Consider a finite-dimensional system in a Gibbs state with its Hamiltonian $H_\lambda$ parametrized by a temperature-independent parameter $\lambda$ at temperature $T$. Let $\xi\in\{\lambda, T\}$ denote an unknown parameter to be estimated. If $\mathcal{F}_{AB}(\xi;T)$ is the global QFI and $\mathcal{F}_{A\to B}(\xi;T)$ is the LOCC QFI for estimating $\xi$, in the high-temperature limit, we have
\begin{equation}
    -\partial_\xi^2\tilde{D}_{A\to B}(\xi;T)=\partial_T\Big(T\Delta\mathcal{F}(\xi;T)\Big)+O(T^{-\alpha_\xi}),
\label{eq:mainresult}
\end{equation}
where $\alpha_\lambda=3$ and $\alpha_T=5$. Particularly, for thermometry ($\xi=T$), $\tilde{D}_{A\to B}(T)$ becomes the diagonal discord $\mathcal{D}_{A\to B}(T)$, which obeys
\begin{equation}
    -\partial_T^2\mathcal{D}_{A\to B}(T)=\partial_T\Big(T\Delta\mathcal{F}(T)\Big)+O(T^{-5}),
\label{eq:temperature}
\end{equation}
\end{theorem}
%todo: I would turn this into theorem+corollary as well...
%\pc{I would turn this into theorem+corollary as well...}

\begin{proof}
First, let us prove the case for $\xi=\lambda$. 
For a general finite-dimensional system, in the high-temperature limit, the state of the total system $\rho_{AB,\lambda}$ can be written as
\begin{align*}
    \rho_{AB,\lambda}=\frac{1}{d_{AB}}\Big(\openone_{AB}-\frac{1}{T}\Big(H_\lambda-\frac{\text{Tr}[H_\lambda]}{d_{AB}}\Big)\Big)+O(T^{-2}),
\end{align*}
where $d_{AB}$ is the dimension of the system. The reduced state of subsystem $A$ is given by $\rho_{A,\lambda}=\text{Tr}_B[\rho_{AB,\lambda}]$, and within the same approximation we have $\rho_{A,\lambda}\propto \openone_A-\frac{1}{T}\big(H_{A,\lambda}+\Omega_{A,\lambda}\big)+O(T^{-2})$, where $\Omega_{A,\lambda}=\text{const}+\frac{1}{d_B}\sum_{k}\langle E_k^{(B)}|H_{AB,\lambda}|E_k^{(B)}\rangle$, which is \textit{independent} of temperature. In the high-temperature limit, $\rho_{A,\lambda}$ can be approximated by a Gibbs state $\rho_{A,\lambda}\simeq \mathcal{Z}_{A,\lambda}^{-1}e^{-H_{A,\lambda}^{\text{eff}}/T}$ with the effective Hamiltonian $H_{A,\lambda}^{\text{eff}}=H_{A,\lambda}+\Omega_{A,\lambda}$ and the normalization factor $\mathcal{Z}_{A,\lambda}=\text{Tr}[e^{-H_{A,\lambda}^{\text{eff}}/T}]$. Then, the local QFI follows Eq.~(\ref{eq:entropy2}), i.e.,
\begin{equation}
    \partial_T\Big(T\mathcal{F}_A(\lambda;T)\Big)=\partial_\lambda^2S_A(\lambda;T)+O(T^{-3}).
    \label{eq:A}
\end{equation}

Suppose that projectors $\tilde{\Pi}_j^A$ are the local optimal projection measurements for estimating $\lambda$ from state $\rho_{A,\lambda}$. Then, the conditional state $\rho_{B|\tilde{\Pi}_j^A,\lambda}$ after measuring subsystem $A$ can also be approximated as a Gibbs state in the high-temperature limit with the effective Hamiltonian 
$H_{B|\tilde{\Pi}_j^A,\lambda}=H_{B,\lambda}+\Omega_{B|\tilde{\Pi}_j^A,\lambda}$, where $\Omega_{B|\tilde{\Pi}_j^A,\lambda}=\text{const}+\text{Tr}[H_{AB,\lambda}\tilde{\Pi}_j^A]$. Then, the local QFI obeys Lemma~\ref{lemma1}, i.e., 
\begin{equation}
    \partial_T\Big(T\mathcal{F}_{B|\tilde{\Pi}_j^A}(\lambda;T)\Big)=\partial_\lambda^2S_{B|\tilde{\Pi}_j^A}(\lambda;T)+O(T^{-3}).
    \label{eq:B}
\end{equation}
Let us select $\tilde{\Pi}_{j*}^{A}$ such that $\sum_j p_{j*}(\lambda;T)S_{B|\tilde{\Pi}_{j*}^{A}}(\lambda; T)=\min_{\tilde{\Pi}_j^A}\sum_j p_{j}(\lambda;T)S_{B|\tilde{\Pi}_j^{A}}(\lambda; T)$. Then, 
\begin{align*}
\begin{split}
    \partial_\lambda^2\tilde{D}_{A\to B}(\lambda;T)=&\Big(\partial_\lambda^2S_A+\sum_{j*} p_{j*}\partial_\lambda^2S_{B|\tilde{\Pi}_{j*}^{A}}-S_{AB}\Big)\\
&+\sum_{j*}\Big(\partial_\lambda^2 p_{j*} S_{B|\tilde{\Pi}_{j*}^{A}}+2\partial_\lambda p_{j*}\partial_\lambda S_{B|\tilde{\Pi}_{j*}^{A}}\Big).
\end{split}
\end{align*}
From Eqs.~(\ref{eq:entropy2}), ~(\ref{eq:A}), and (\ref{eq:B}), we can obtain
\begin{align*}
\begin{split}
    -\partial_\lambda^2\tilde{D}_{A\to B}(\lambda;T)=&\partial_T\Big(T\Delta\mathcal{F}(\lambda;T)\Big)\\
    &-\sum_{j*}\Big(\partial_\lambda^2 p_{j*} S_{B|\tilde{\Pi}_{j*}^{A}}+2\partial_\lambda p_{j*}\partial_\lambda S_{B|\tilde{\Pi}_{j*}^{A}}\Big)
    \end{split}
\end{align*}

In the high-temperature limit, the entropy has the order of
$S_{B|\tilde{\Pi}_{j*}^A}(\lambda;T)=\ln (d_B)+O(T^{-2})$ and the measurement probability is 
\begin{equation}
    p_{j*}(\lambda;T)=\text{Tr}[(\tilde{\Pi}_{j*}^{A}\otimes\openone_B)\rho_{AB,\lambda}(\tilde{\Pi}_{j*}^{A\dagger}\otimes\openone_B)]=\frac{1}{d_A}+O(T^{-1}). 
\label{eq:probability}
\end{equation}

In the high-temperature limit, we have
\begin{align*}
    \partial_\lambda^2S_{B|\tilde{\Pi}_{j*}^A}(\lambda;T)&=O(T^{-2}).
\end{align*}
By using the fact that $\sum_{j*}p_{j*}(\lambda;T)=1$, we can write
\begin{align*}
\begin{split}
    \sum_{j*}\partial_\lambda^2p_{j*}(\lambda;T)S_{B|\tilde{\Pi}_{j*}^A}(\lambda;T)&=O(T^{-1})O(T^{-2})=O(T^{-3})\\
    \sum_{j*}\partial_\lambda p_{j*}\partial_\lambda S_{B|\tilde{\Pi}_{j*}^{A}}&=O(T^{-1})O(T^{-2})=O(T^{-3}).
    \end{split}
\end{align*}
Therefore, we can write
\begin{align*}
    -\partial_\lambda^2\tilde{D}_{A\to B}(\lambda;T)=\partial_T\Big(T\Delta\mathcal{F}(\lambda;T)\Big)+O(T^{-3}). 
\end{align*}

Second, for thermometry, from Lemma~\ref{lemma2} and Definition~\ref{definition}, the optimal measurement basis is the diagonal basis of $\rho_{A,T}$. Therefore, discord for local metrology $\tilde{D}_{A\to B}(T)$ becomes diagonal discord $\mathcal{D}_{A\to B}(T)$.  From our previous result in Ref.~\cite{Sone18a} since we have already known that
\begin{align*}
-\frac{1}{T}\partial_T\mathcal{D}_{A\to B}(T)=\Delta\mathcal{F}(T)+O(T^{-5}),
\end{align*}
we can obtain
\begin{align*}
    -\partial_T^2\mathcal{D}_{A\to B}(T)=\partial_T\Big(T\Delta\mathcal{F}(T)\Big)+O(T^{-5}),
\end{align*}

\end{proof}

Therefore, for any parameter $\xi$, in the high-temperature limit, we can approximately write 
\begin{equation}
\partial_\xi^2\tilde{D}_{A\to B}(\xi;T)\simeq -\partial_T\big(T\Delta\mathcal{F}(\xi;T)\big),
\label{approx}
\end{equation}
which demonstrates that $\partial_T\big(T\Delta\mathcal{F}(\lambda;T)\big)$ is the curvature of $\tilde{D}_{A\to B}$. %\as{Here, note that this relation is coming from the relations between the QFI and entropy in the thermodynamical system, which are shown in Eq.~(\ref{eq:entropy2}) and Eq.~(\ref{eq:entropy3}).} 
{Even if the curvature of the discord for local metrology is not directly related to the amount of nonclassical correlations, Eq.~(\ref{approx}) still describes the role of  nonclassical correlations  in the greedy local measurement scheme in the LOCC regime.}
{Although we derived Theorem~\ref{theorem} for a bipartite system,  the results in the high-temperature limit can be extended to the case of multipartite systems (see Appendix~\ref{app:multipartite}).}

When the parameter $\lambda$ is linearly coupled in the Hamiltonian, discord for local metrology becomes diagonal discord. From Theorem~\ref{theorem} and Lemma~\ref{lemma2}, we can obtain the following corollary: 
\begin{corollary}
\label{corollary3}
Consider a finite-dimensional system in a Gibbs state at temperature $T$ with its Hamiltonian parametrized by a temperature-independent parameter $\lambda$. When $\lambda$ is linearly coupled to the Hamiltonian $H_\lambda$, i.e., $\partial_\lambda^2 H_\lambda=0$, we have
\begin{equation}
    \partial_\lambda^2\mathcal{D}_{A\to B}(\lambda;T)=-\partial_T\big(T\Delta\mathcal{F}(\lambda;T)\big)+O(T^{-3}),
    \label{eq:DDFisher}
\end{equation}
where $\mathcal{D}_{A\to B}(\lambda;T)$ is the diagonal discord. 
\end{corollary}

In addition, let us note the case of estimating a parameter linearly coupled to the single-body term. For this case, we can obtain the following corollary (see Appendix~\ref{app:corollary4} for proof):
\begin{corollary}
\label{corollary4}
For a finite-dimensional system in a Gibbs state at temperature $T$, when $\lambda$ is a parameter linearly coupled to the single-body term as
\begin{align*}
    H_\lambda=\lambda H_A+\lambda H_B +H_{AB},
\end{align*}
where $H_A$ and $H_B$ are the system Hamiltonians and $H_{AB}$ is the interaction Hamiltonian,  then, we have
\begin{equation}
\begin{split}
    -\partial_\lambda^2\mathcal{D}_{A\to B}(\lambda;T)&=O(T^{-3})\\
    \partial_T\Big(T\Delta\mathcal{F}(\lambda;T)\Big)&=O(T^{-3}). 
\end{split}
\label{eq:localfieldmagnitude}
\end{equation}
\end{corollary}
To this order, the local measurements are optimal. Here, note that the leading term that Theorem~\ref{theorem} cares about is $O(T^{-2})$, and this  is $0$ in this case.

In the following section, we show some examples that verify  Theorem~\ref{theorem}, Corollary~\ref{corollary3}, and \ref{corollary4}.

\section{Examples}
\label{sec:example}
In this section, we verify the relation in Eqs.~(\ref{eq:mainresult}), (\ref{eq:temperature}), (\ref{eq:DDFisher}), and (\ref{eq:localfieldmagnitude}) by providing several examples of two-qubit Heisenberg interaction, whose Hamiltonian can be written as 
\begin{align*}
H=\frac{B_1}{2}Z_A+\frac{B_2}{2}Z_B+\frac{J_x}{2}X_AX_B+\frac{J_y}{2}Y_AY_B+\frac{J_z}{2}Z_AZ_B,
\end{align*}
where $X_j, Y_j$, and $Z_j$ $(j=A,~B)$ are the Pauli matrices acting on the $j$th spin.  

\subsection{Thermometry}
\label{sec:thermometry}
First, let us discuss the case of thermometry. From our recent result~\cite{Sone18a}, we have:
\begin{align*}
\begin{split}
&\Delta\mathcal{F}(T)=\frac{J_x^2+J_y^2}{4T^4}+O(T^{-5})\\
&-\frac{1}{T}\partial_T\mathcal{D}_{A\to B}(T)=\frac{J_x^2+J_y^2}{4T^4}+O(T^{-5}),
\end{split}    
\end{align*}
which directly yields
\begin{align*}
\begin{split}
    \partial_T\big(T \Delta\mathcal{F}(T)\big)&=-\frac{3(J_x^2+J_y^2)}{4T^{4}}+O(T^{-5})\\
    -\partial_T^2\mathcal{D}_{A\to B}(T)&=-\frac{3(J_x^2+J_y^2)}{4T^{4}}+O(T^{-5}).
\end{split}
\end{align*}
Therefore, Eq.~(\ref{eq:temperature}) is valid. 

\subsection{Coupling strength}
\label{sec:coupling}
Next, let us consider the case of estimating the coupling strength $J$ when $J_x=J_y=J$.
Then, we have
\begin{align*}
\begin{split}
    &\Delta\mathcal{F}(J; T)=\frac{1}{2T^2}+O(T^{-3})\\
    &\mathcal{D}_{A\to B}(J; T)=\frac{J^2}{4T^2}+O(T^{-3}),
    \end{split}
\end{align*}
which directly yields
\begin{align*}
\begin{split}
    \partial_T\big(T \Delta\mathcal{F}(J; T)\big)&=-\frac{1}{2T^2}+O(T^{-3})\\
    -\partial_J^2\mathcal{D}_{A\to B}(J; T)&=-\frac{1}{2T^2}+O(T^{-3}),
\end{split}
\end{align*}
Therefore, Eq.~(\ref{eq:DDFisher}) is valid. 

\subsection{Magnetometry}
\label{sec:magnetic}
Finally, let us consider magnetometry, which demonstrates Eq.~(\ref{eq:localfieldmagnitude}). We consider the case of $B_1=B_2=B$, where $B$ is the parameter to be estimated. In this case, we can find that
\begin{align*}
    \begin{split}
    \partial_T \big(T \Delta\mathcal{F}(B; T)\big)&=-\frac{(J_x-J_y)^2}{8T^4}+O(T^{-5})\\
    -\partial_B^2\mathcal{D}_{A\to B}(B; T)&=-\frac{J_x^2+J_xJ_y+J_y^2}{24T^4}+O(T^{-5}).
    \end{split}
\end{align*}
From Eq.~(\ref{eq:localfieldmagnitude}), the leading term should be $O(T^{-3})$; therefore, we can say that Eq.~(\ref{eq:localfieldmagnitude}) is valid, but the term to the corresponding order $O(T^{-3})$ is $0$.

\section{Conclusion}
\label{sec:conc}

In conclusion, we introduced a metric for nonclassical correlations, the discord for local metrology, which is defined as a quantum discord in the greedy local measurement scheme, and we derived a relation between the discord for local metrology and {the difference between the QFI} of the global optimal scheme and the greedy local measurement scheme in the high-temperature limit. We demonstrated that {the curvature of} the discord for local metrology quantifies the precision loss in the estimation of a general parameter due to availability of local measurements only (Theorem~\ref{theorem}). %By considering optimal local measurements sequentially performed on the subsystems in a general finite-dimensional system, the discord for local metrology thus bounds the optimal achievable sensitivity with local measurements. 
This also indicates that variations in  nonclassical correlations at thermal equilibrium,  quantified by discord for local metrology, {are related to the ability of  the greedy local measurement scheme to achieve the ultimate estimation precision limit, quantified by the global QFI.} %in enhancing the estimation sensitivity. 
We also showed that discord for local metrology coincides with diagonal discord when one estimates a linear coupling parameter (Corollaries~\ref{corollary3} and \ref{corollary4}). 

Although we focused on finite-dimensional systems in the high-temperature limit, it would be interesting to extend the relation between the discord for local metrology and QFI for more general Gibbs states, especially in the low-temperature limit where one could search for connections to phase transition phenomena, {or for infinite-dimensional systems, such as bosonic gases~\cite{Ugo13, Ugo15}. } 

%The direct 
{The relation between the curvature of the discord for local metrology and the difference in the QFI} explicitly demonstrates the role of nonclassical correlations in quantum metrology based on the original definition of quantum discord. %From a fundamental perspective, 
This provides insight on the role of nonclassicality in quantum metrology and motivates further exploration in more general settings, {which can potentially inspire experimentalists to design measurement and control protocols to utilize quantum discord as a resource to achieve precise sensing and imaging, e.g., in the context of room-temperature nuclear magnetic resonance or bioimaging with defect spins~\cite{DeVience15,Sushkov14l, Hui10, Yeung10d}. }
%can provide an insight to explore further profound research  in more general approach. 
%From a more practical point of view, our results in the high-temperature regime could allow experimentalists to design a measurement and control protocol that exploits quantum discord as a resource to achieve precise sensing and imaging, e.g.,  in the context of room-temperature nuclear magnetic resonance or bioimaging with defect spins~\cite{DeVience15,Sushkov14l, Hui10, Yeung10d}.  

\acknowledgments
This work was supported, in part, by the U.S. Army Research Office through Grants No. W911NF-11-1-0400 and Bo.  W911NF-15-1-0548 and by the NSF Grant No. PHY0551153.
{A.S. acknowledges a Thomas G. Stockham Jr. Fellowship.}  
Q.Z. acknowleges the U.S. Department of Energy through Grant No. PH-COMPHEP-KA24 and the Claude E. Shannon Research Assistantship. 
We thank B. Yadin, K. Modi, and R. Takagi for helpful discussions.
%A.S. and Q.Z. contributed equally to this work. 

\appendix

\section{Proof of Lemma~\ref{lemma1}}
\label{app:lemma1}
First, let us prove the case of $\xi=\lambda$.

Let $\epsilon$ be an error in our estimation. Then, the Hamiltonian with the error becomes
\begin{align*}
    H_{\lambda+\epsilon}=H_\lambda+\epsilon G_\lambda+O(\epsilon^2),
\end{align*}
where 
\begin{align*}
G_\lambda=\partial_\lambda H_\lambda.
\end{align*}
The fidelity between $\rho_\lambda$ and $\rho_{\lambda+\epsilon}$ is defined as
\begin{align*}
    \mathbb{F}[\rho_\lambda, \rho_{\lambda+\epsilon}]
    =\Big(\text{Tr}\Big[\sqrt{\rho_{\lambda}^{1/2}\rho_{\lambda+\epsilon}\rho_\lambda^{1/2}}~\Big]\Big)^2.
\end{align*}
Since
\begin{align*}
    e^{-\frac{H_\lambda}{2T}}e^{-(H_\lambda+\epsilon G_\lambda)/T}e^{-\frac{H_\lambda}{2T}}=e^{-(2H_\lambda+\epsilon G_\lambda)/T+O(T^{-3})}
\end{align*}
we can write
\begin{align*}
\begin{split}
    \mathbb{F}[\rho_\lambda,\rho_{\lambda+\epsilon}]&=\frac{1}{\mathcal{Z}_\lambda\mathcal{Z}_{\lambda+\epsilon}}\Big(\text{Tr}[e^{-(H_\lambda+\frac{\epsilon}{2}G_\lambda)/T+O(T^{-3})}]\Big)^2\\
    &=\frac{1}{\mathcal{Z}_\lambda\mathcal{Z}_{\lambda+\epsilon}}\Big(\text{Tr}[e^{-(H_\lambda+\frac{\epsilon}{2}G_\lambda)/T}]\Big)^2+O(T^{-3}).
\end{split}
\end{align*}

In the high-temperature limit, the fidelity between $\rho_\lambda$ and $\rho_{\lambda+\epsilon}$ becomes
\begin{align*}
    \mathbb{F}[\rho_\lambda,\rho_{\lambda+\epsilon}]= \frac{\mathcal{Z}_{\lambda+\frac{\epsilon}{2}}^2}{\mathcal{Z}_\lambda\mathcal{Z}_{\lambda+\epsilon}}+O(T^{-3}),
\end{align*}
where $\mathcal{Z}_{\lambda+\frac{\epsilon}{2}}=\text{Tr}[e^{- H_{\lambda+\frac{\epsilon}{2}}/T}]$, and from the definition of  the QFI, we can obtain
\begin{align*}
\mathcal{F}(\lambda;T)=\frac{\mathcal{Z}_\lambda\partial_\lambda^2\mathcal{Z}_\lambda-(\partial_\lambda\mathcal{Z}_\lambda)^2}{\mathcal{Z}_\lambda^2}+O(T^{-3})
\end{align*}
Here, for the Gibbs state, $\langle G_\lambda\rangle=\text{Tr}[G_\lambda\rho_\lambda]$ is always
\begin{align*}
    \langle G_\lambda\rangle=-T\partial_\lambda \ln\mathcal{Z}_\lambda.
\end{align*}
Then, the susceptibility with respect to a temperature-independent parameter $\lambda$ can be defined  as 
\begin{align*}
    \chi(\lambda;T)=-\partial_\lambda\langle G_\lambda\rangle. 
\end{align*}
so that we have 
\begin{equation}
    \mathcal{F}(\lambda;T)=\frac{\chi(\lambda;T)}{T}+O(T^{-3}).
    \label{eq:QFI}
\end{equation}
Since the entropy of the bipartite system, $S(\lambda;T)=-\text{Tr}[\rho_\lambda\ln\rho_\lambda]$,  
satisfies the following relation
\begin{equation}
\partial_T\langle G_\lambda\rangle=-\partial_\lambda S(\lambda;T),
\label{eq:entropy}
\end{equation}
from Eqs.~(\ref{eq:QFI}) and Eq.~(\ref{eq:entropy}), we can obtain
\begin{align*}
\partial_T\big(T\mathcal{F}(\lambda;T)\big)=\partial_\lambda^2 S(\lambda;T)+O(T^{-3}). 
\end{align*}

Second, for the thermometry case $\xi=T$, the global QFI is $\mathcal{F}(T)=C(T)/T^2$~\cite{correa2015individual} for finite temperature, where $C(T)$ is the heat capacity so that $C(T)=T\partial_TS(T)$. Therefore, we can obtain an exact relation
\begin{align*}
\partial_T^2S(T)=\partial_T\big(T\mathcal{F}(T)\big).
\end{align*} 

Therefore, Lemma~\ref{lemma1} is valid.

\section{Proof of Lemma~\ref{lemma2}}
\label{app:lemma2}
First, let us prove the case $\xi=\lambda$. 
When $\partial_\lambda G_\lambda=0$, the QFI becomes
\begin{equation}
    \mathcal{F}(\lambda,T)=\frac{(\delta G_\lambda)^2}{T^2}+O(T^{-3}).
    \label{eq:linear}
\end{equation}
Let $E_k(\lambda)$ be the eigenvalues of the Hamiltonian $H_\lambda$. Then, $H_\lambda$ can be diagonalized as $H_\lambda=P_\lambda K_\lambda P_\lambda^\dagger$,
where $P_\lambda$ is an unitary operator, $P_\lambda^{\dagger}P_{\lambda}=P_\lambda P_\lambda^{\dagger}=\openone$
and $K_\lambda=\text{diag}(E_1(\lambda),E_2(\lambda),\cdots, E_{d}(\lambda))=\sum_{k=1}^{d}E_k(\lambda)|k\rangle\langle k|$
and $|k\rangle$'s form a complete basis independent of $\lambda$,  
and $d$ is the dimension of the system. Thus, 
\begin{align*}
\partial_\lambda K_\lambda=\sum_{k=1}^d\partial_\lambda E_k(\lambda)|k\rangle\langle k|.
\end{align*}
Then, the Gibbs state becomes 
\begin{align*}
    \rho_\lambda=\frac{1}{\mathcal{Z}_\lambda}P_\lambda e^{-K_\lambda/T}P_\lambda^\dagger=\frac{1}{\mathcal{Z}_\lambda}\sum_{k=1}^d e^{-E_k(\lambda)/T}P_\lambda|k\rangle\langle k|P_\lambda^\dagger.
\end{align*}
%\begin{align*}
%\rho_\lambda=\mathcal{Z}_\lambda^{-1}e^{-H_\lambda/T}=\mathcal{Z}_\lambda^{-1}P_\lambda e^{-K_\lambda/T}P_\lambda^\dagger=\mathcal{Z}_\lambda^{-1}\sum_{k=1}^d e^{-E_k(\lambda)/T}P_\lambda|k\rangle\langle k|P_\lambda^\dagger.
%\end{align*}

Let us calculate the expectation value of $G_\lambda=\partial_\lambda H_\lambda$. Since 
\begin{align*}
G_\lambda=\partial_\lambda P_\lambda K_\lambda P_{\lambda}^{\dagger}+P_\lambda\partial_\lambda K_\lambda P_\lambda^\dagger+P_\lambda K_\lambda\partial_\lambda(P_\lambda^\dagger),
\end{align*}
we have
\begin{align*}
\begin{split}
\langle G_\lambda\rangle=&\text{Tr}[\rho_\lambda G_\lambda]\\
=&\frac{1}{\mathcal{Z}_\lambda}\text{Tr}\Big[\Big(P_\lambda e^{-K_\lambda/T}P_\lambda^\dagger\Big)\Big(\partial_\lambda P_\lambda K_\lambda P_{\lambda}^\dagger\\
&+
P_\lambda\partial_\lambda K_\lambda P_\lambda^\dagger+P_\lambda K_\lambda\partial_\lambda(P_\lambda^\dagger)\Big)\Big]\\
=&\text{Tr}\Big[\frac{e^{-K_\lambda/T}}{\mathcal{Z}_\lambda}\partial_\lambda K_\lambda\Big]+\frac{1}{\mathcal{Z}_\lambda}\text{Tr}\Big[K_\lambda e^{-K_\lambda /T}(P_\lambda^\dagger\partial_\lambda P_\lambda)\\
&+
(\partial_\lambda(P_\lambda^\dagger)P_\lambda) e^{-K_\lambda/T}K_\lambda \Big]\\
=&\text{Tr}\Big[\frac{e^{-K_\lambda/T}}{\mathcal{Z}_\lambda}\partial_\lambda K_\lambda\Big]+\frac{1}{\mathcal{Z}_\lambda}\text{Tr}\Big[e^{-K_\lambda/T}K_\lambda \partial_\lambda(P_\lambda^\dagger P_\lambda)\Big]\\
=&\text{Tr}\Big[\frac{e^{-K_\lambda/T}}{\mathcal{Z}_\lambda}\partial_\lambda K_\lambda\Big]=\text{Tr}[\rho_\lambda P_\lambda\partial_\lambda K_\lambda P_\lambda^\dagger],
\end{split} 
\end{align*}
where we used the cyclic property of trace operation and the fact that $[e^{-K_\lambda/T},K_\lambda]=0$.
Therefore, 
\begin{align*}
\langle G_\lambda\rangle=\langle P_\lambda\partial_\lambda K_\lambda P_\lambda^\dagger \rangle,
\end{align*}
and $P_\lambda\partial_\lambda K_\lambda P_\lambda^\dagger$ has same diagonal basis of $\rho_\lambda$, which is $\{P_\lambda|k\rangle\langle k|P_\lambda^\dagger\}_{k=1}^{d}$. This means that the optimal measurement for estimating the linear coupling parameter is the projection measurement to the diagonal basis of $\rho_\lambda$. 

Second, for the case of $\xi=T$, the QFI is given as
\begin{align*}
    \mathcal{F}(T)=\frac{C(T)}{T^2},
\end{align*}
where $C(T)$ is the heat capacity~\cite{correa2015individual}. Because of $C(T)=\partial_T\langle H_\lambda\rangle=(\delta H_\lambda)^2/T^2$, the temperature variance $(\delta T)^2$ becomes
\begin{align*}
    (\delta T)^2=\frac{(\delta H_\lambda)^2}{(\partial_T\langle H_\lambda \rangle)^2}=\frac{T^2}{C(T)}=\frac{1}{\mathcal{F}(T)}.
\end{align*}
Therefore, for  thermometry, the projection measurements on diagonal basis are optimal.

\section{Proof of Corollary~\ref{corollary4}}
\label{app:corollary4}

Let us consider the following Hamiltonian:
\begin{align*}
    H_\lambda=\lambda H_A+\lambda H_B+H_{AB},
\end{align*}
where $H_A$ and $H_B$ are the system Hamiltonians, i.e., $[H_A,H_B]=0$ and $H_{AB}$ is the interaction Hamiltonian and generally $[H_A+H_B, H_{AB}]\neq 0$. Here, $\lambda$ is the parameter to be estimated. In this case, $H_{\lambda+\epsilon}=H_\lambda+\epsilon G_\lambda+O(\epsilon^2)$, where $G_\lambda=H_A+H_B$, which is independent of $\xi=\{\lambda,T\}$. Here, we just simply write $G_\lambda$ as $G$ in order to emphasize its independence of $\lambda$.

We  already know that for the Gibbs state, we have $\langle G\rangle=-T\partial_\lambda \ln\mathcal{Z}_\lambda$.
In this case, we can immediately obtain 
\begin{align*}
\langle G\rangle=\langle H_A\rangle+\langle H_B\rangle=O(T^{-1})
\end{align*}
because the entropy is $S_{AB}(\lambda; T)=\ln (d_{AB})+O(T^{-2})$ and the relation between the entropy and $\langle G\rangle$ is
\begin{align*}
\partial_T \langle G\rangle=-\partial_\lambda S_{AB}(\lambda;T)=O(T^{-2}).
\end{align*}
By defining a general susceptibility with respect to $\lambda$ as
\begin{align*}
\chi(\lambda;T)=-\partial_\lambda\langle G\rangle=O(T^{-1}),
\end{align*}
the QFI can be given as
\begin{equation}
\mathcal{F}_{AB}(\lambda; T)=-\frac{1}{T}\partial_\lambda\langle H_A\rangle-\frac{1}{T}\partial_\lambda\langle H_B\rangle=O(T^{-2}).
\label{eq:localfieldFAB}
\end{equation}

Now, let us consider the subsystem A. The effective Hamiltonian $H_{A,\lambda}^{\text{eff}}$ can be written as
$H_{A,\lambda}^{\text{eff}}=\lambda (H_A+\text{const})+\Omega_A$.
Therefore, 
\begin{align*}
    \mathcal{F}_A(\lambda;T)=-\frac{1}{T}\partial_\lambda\langle H_A\rangle+O(T^{-3}).
\end{align*}
Similarly, for $\rho_{B|\tilde{\Pi}_{j*}^A,\xi}$, we have
\begin{align*}
    \mathcal{F}_{B|\tilde{\Pi}_{j*}^A}(\lambda;T)=-\frac{1}{T}\partial_\lambda\langle H_B\rangle+O(T^{-3}).
\end{align*}
Therefore, by using Eq.~(\ref{eq:probability}), we have 
\begin{align*}
    \mathcal{F}_{A\to B}(\lambda;T)=-\frac{1}{T}\Big(\partial_\lambda\langle H_A\rangle+\partial_\lambda\langle H_B\rangle\Big)+O(T^{-3})
    \end{align*}
From Eqs.~(\ref{eq:localfieldFAB}) and Eq.~(\ref{eq:DDFisher}), we can obtain
\begin{align*}
\begin{split}
    -\partial_\lambda^2\mathcal{D}_{A\to B}(\lambda;T)&=O(T^{-3})\\
    \partial_T\Big(T\Delta\mathcal{F}(\lambda;T)\Big)&=O(T^{-3}). 
\end{split}
\end{align*}

\section{{Generalization to the multipartite case}}
\label{app:multipartite}
Let us consider a finite-dimensional system composed of $N$ subsystems indexed by integers $1\le k\le N$. In the multipartite case, each subsystem is measured with local optimal measurement sequentially, and we demonstrate that the difference in global QFI and LOCC QFI can be quantified via the curvature of the discord for local metrology in the high-temperature limit, in parallel to Ref.~\cite{Sone18a}. 

We denote the order of measurement in a greedy local measurement scheme by $\sigma_{1:N}\equiv (\sigma_1,\sigma_2,\cdots,\sigma_N)$, where $\sigma_k=\{1,2,\cdots,N\}$. Let us write $\mathcal{H}(\sigma_k)$ as the Hilbert space of the system on which we perform local optimal measurement $\tilde{\Pi}_{\sigma_k}$ and $\mathcal{H}(\sigma_{k+1:N})$ as the Hilbert space of the rest of system on which we perform the local optimal measurement $\tilde{\Pi}_{\sigma_{k+1:N}}$. Therefore, the total system can be decomposed sequentially into
\begin{align*}
    \begin{split}
        \mathcal{H}(\sigma_{1:N})&=\mathcal{H}(\sigma_1)\otimes\mathcal{H}(\sigma_{2:N})\\
        &=\mathcal{H}(\sigma_1)\otimes\mathcal{H}(\sigma_2)\otimes\mathcal{H}(\sigma_{3:N})\\
        &~\vdots\\
        &=\mathcal{H}(\sigma_1)\otimes\mathcal{H}(\sigma_2)\otimes\cdots\otimes\mathcal{H}(\sigma_{k})\otimes\mathcal{H}(\sigma_{k+1:N}),
    \end{split}
\end{align*}
where $2\le k\le N-1$. 

In the first step $(k=1)$, we first perform the local optimal measurement $\tilde{\Pi}_{\sigma_1}$.  Then conditioned on the measurement result of $\tilde{\Pi}_{\sigma_1}$, we perform the other local optimal measurement $\tilde{\Pi}_{\sigma_{2:N}}$ on the rest of system. Let us write the global QFI as $\mathcal{F}_{\sigma_{1:N}}$ and LOCC QFI as $\mathcal{F}_{\sigma_1\to\sigma_{2:N}}$. Then, in the high-temperature limit, from Eq.~(\ref{approx}), we have 
\begin{align*}
\partial_T\Big(T( \mathcal{F}_{\sigma_{1:N}}-\mathcal{F}_{\sigma_1\to\sigma_{2:N}})\Big)\simeq-\partial_\xi^2\tilde{D}_{\sigma_1\to\sigma_{2:N}}.
\end{align*}
For the $2\le k \le N-1$ steps, the measurement $\tilde{\Pi}_{\sigma_k}$ is conditioned on the results of the previous sequence of local optimal measurements $\tilde{\Pi}_{1:k-1}\equiv (\tilde{\Pi}_{\sigma_{1}},\tilde{\Pi}_{\sigma_2},\cdots,\tilde{\Pi}_{\sigma_{k-1}})$. We treat the rest of system as a bipartite system composed of $\mathcal{H}(\sigma_k)$ and $\mathcal{H}(\sigma_{k+1:N})$. Then, from Eq.~(\ref{approx}), we have 
$\partial_T\Big(T(\mathcal{F}_{\sigma_{k:N}|\tilde{\Pi}_{\sigma_{1:k-1}}}-\mathcal{F}_{\sigma_{k}\to\sigma_{k+1:N}|\tilde{\Pi}_{\sigma_{1:k-1}}})\Big)\simeq-\partial_\xi^2\tilde{D}_{\sigma_{k}\to\sigma_{k+1:N}|\tilde{\Pi}_{\sigma_{1:k-1}}}$. Here, we have $\mathcal{F}_{\sigma_{k}\to\sigma_{k+1:N}|\tilde{\Pi}_{\sigma_{1:k-1}}}=\mathcal{F}_{\sigma_{k}|\tilde{\Pi}_{\sigma_{1:k-1}}}+\mathcal{F}_{\sigma_{k+1:N}|\tilde{\Pi}_{\sigma_{1:k}}}$.  

The unconditional QFI is given by the average over measurement outcome distribution $p(\tilde{\Pi}_{\sigma_{1:k-1}})$ as
\begin{align*}
\mathcal{F}_{\sigma_k\to\sigma_{k+1:N}|\sigma_{1:k-1}}\equiv
\sum_{\tilde{\Pi}_{\sigma_{1:k-1}}} p(\tilde{\Pi}_{\sigma_{1:k-1}})\mathcal{F}_{\sigma_k\to\sigma_{k+1:N}|\tilde{\Pi}_{\sigma_{1:k-1}}}.
\end{align*}
Then one can define an unconditional version of discord
\begin{align*}
    \tilde{D}_{\sigma_k\to\sigma_{k+1:N}|\sigma_{1:k-1}}=\sum_{\tilde{\Pi}_{\sigma_{1:k-1}}} p(\tilde{\Pi}_{\sigma_{1:k-1}})\tilde{D}_{\sigma_k\to\sigma_{k+1:N}|\tilde{\Pi}_{\sigma_{1:k-1}}},
\end{align*}
which is related to the average measurement precision difference,
\begin{align*}
\begin{split}
    \partial_T\Big(T(\mathcal{F}_{\sigma_{k:N}|\sigma_{1:k-1}}-&\mathcal{F}_{\sigma_{k}\to\sigma_{k+1}|\sigma_{1:k-1}})\Big)\\
    &\simeq-\partial_\xi^2\tilde{D}_{\sigma_k\to\sigma_{k+1:N}|\sigma_{1:k-1}},
    \end{split}
\end{align*}
where $\mathcal{F}_{\sigma_{k}\to\sigma_{k+1}|\sigma_{1:k-1}}=\mathcal{F}_{\sigma_{k}|\sigma_{1:k-1}}+\mathcal{F}_{\sigma_{k+1:N}|\sigma_{1:k}}$. Therefore, by adding the equation above from $k=1$ and $k=N$, the difference in the QFI can be written as
\begin{align*}
    \Delta\mathcal{F}_{\sigma_{1:N}}=\mathcal{F}_{\sigma_{1:N}}-\sum_{k=1}^N\mathcal{F}_{\sigma_{\sigma_k}|\sigma_{1:k-1}}
\end{align*}
so that we can obtain
\begin{equation}
   \partial_T(T\Delta\mathcal{F}_{\sigma_{1:N}})\simeq-\partial_\xi^2\tilde{D}_{\sigma_{1:N}},
   \label{eq:multipartite}
\end{equation}
where
\begin{align*}
  \tilde{D}_{\sigma_{1:N}}=\sum_{k=1}^N \tilde{D}_{\sigma_k\to\sigma_{k+1:N}|\sigma_{1:k-1}}.  
\end{align*}
Equation~(\ref{eq:multipartite}) is the generalization of Eq.~(\ref{approx}) for the multipartite case. 

\bibliographystyle{apsrev4-1}
\bibliography{Biblio}

\end{document}